\documentclass[12pt]{article}

\usepackage{cmap} %
\usepackage[utf8]{inputenc}
\usepackage[english]{babel}
\usepackage[T1]{fontenc}
\usepackage{appendix}
\usepackage{braket}
\usepackage{multirow}
\usepackage{url, amsfonts, tikz, physics, listings, xcolor, graphicx, float}
\usepackage[shortlabels]{enumitem}
\usepackage{dsfont}
\usepackage{amsmath, amssymb, amstext, amscd, amsthm, makeidx, graphicx, url, mathrsfs, mathtools, longdivision, polynom, bbm, complexity}
\usepackage{nicefrac}
\usepackage[linesnumbered,ruled,vlined]{algorithm2e}
\usepackage{xcolor}
\usepackage[normalem]{ulem}
\usepackage{booktabs}
\usepackage{comment}

\newcommand{\myparagraph}[1]{\smallskip \noindent \textbf{#1}. }

\usepackage{blindtext}
\usepackage{titlesec}

\definecolor{blueviolet}{rgb}{0.2, 0.2, 0.6}
\definecolor{webgreen}{rgb}{0,.5,0}
\definecolor{webbrown}{rgb}{.6,0,0}
\usepackage[pdftex,
  bookmarks=false,
  colorlinks=true, %
  urlcolor=webbrown,
  linkcolor=blueviolet, 
  citecolor=webgreen,
  pdfstartpage=1,
  pdfstartview={FitH},  %
  bookmarksopen=false
  ]{hyperref}
\usepackage{cleveref}

\usepackage[margin=0.8in]{geometry}
\newcommand{\Gen}{\mathsf{Gen}}
\newcommand{\Rec}{\mathsf{Rec}}
\newcommand{\Eval}{\mathsf{Eval}}
\newcommand{\Adv}{\mathsf{Adv}}
\newcommand{\negl}{\mathsf{negl}}

\newtheorem{theorem}{Theorem}
\newtheorem{definition}{Definition}
\newtheorem{proposition}{Proposition}
\newtheorem{lemma}{Lemma}
\newtheorem*{lem*}{Lemma}

\newtheorem{remark}{Remark}
\newtheorem*{rem*}{Remark}
\newtheorem{corollary}{Corollary}

\newtheorem*{fact*}{Fact}

\newcommand{\footremember}[2]{%
    \footnote{#2}
    \newcounter{#1}
    \setcounter{#1}{\value{footnote}}%
}

\begin{document}

\title{Quantum function secret sharing}

\author{
Alex B. Grilo\footremember{cnrs}{Sorbonne Université, CNRS, LIP6, Paris, France} 
\and
Ramis Movassagh\footremember{google}{Google Quantum AI, Venice, CA, 90291, U.S.A.
}}

\date{}

\maketitle

\begin{abstract}
We propose a quantum function secret sharing scheme in which the communication is exclusively classical. In this primitive, a classical dealer distributes a secret quantum circuit $C$ by providing shares to  $p$ quantum parties. The parties on an input state $\ket{\psi}$ and a projection $\Pi$, compute values $y_i$ that they then classically communicate back to the dealer, who can then compute $\norm{\Pi C\ket{\psi}}^2$ using only classical resources. Moreover,  the shares do not leak much information about the secret circuit $C$.
Our protocol for quantum secret sharing uses the {\em Cayley path}, a tool that has been extensively used to support quantum primacy claims. More concretely, the shares of $C$ correspond to randomized version of $C$ which are delegated to the quantum parties, and the reconstruction can be done by extrapolation. Our scheme has two limitations, which we prove to be inherent to our techniques: First, our scheme is only secure against single adversaries, and we show that if two parties collude, then they can break its security. Second, the evaluation done by the parties requires exponential time in the number of gates.

\end{abstract}

\section{Introduction and summary of results}
Universal blind quantum computing~\cite{BroadbentFK09} was a breakthrough result in quantum cryptography allowing a client to delegate a quantum circuit of their choice to a quantum server {\em without revealing the circuit}. This is even more impressive when we realize that $1)$ their result is information-theoretically secure, and $2)$ it predates classical protocols for fully-homomorphic encryption~\cite{Gentry09}, which requires computational assumptions. However, there is a crucial point in the security of the protocol: the client needs to send random qubit states to the servers, so the client is not fully classical. This was later shown to be necessary to achieve universal blind quantum computing with information-theoretical security with a single server~\cite{AaronsonCGK19}.

While other results achieve blind quantum computing in the multi-server setting with fully classical clients~\cite{cgjv24},  the servers must share highly entangled states. It is thus an open question of whether it is possible to achieve schemes for universal blind quantum computing with multiple servers and classical clients, where the servers are ``stand alone'' parties that do not share entanglement. Such a question turns out to be a quantum generalization of what is called {\em function secret sharing} in classical cryptography.

In this work, we define quantum function secret sharing, propose a protocol for it, and discuss its advantages and limitations. We now summarize our contributions in more details.\\

\myparagraph{Defining quantum function secret sharing.} Classical secret sharing is a fundamental cryptographic primitive~\cite{shamir1979share,blakley1979safeguarding} where a \emph{dealer} is able to share a secret bitstring $s$ among  $p$ parties, where only allowed subsets of parties are able to recover the secret.
One of the most common settings for secret sharing is threshold secret sharing. In this setting given a fixed threshold $t$, each party receives a bitstring $s_i$ from the dealer, and given $(s_i)_{i \in T}$, the dealer retrieves the original $s$ if and only if (iff) $|T| \geq t$.
Such schemes have numerous applications such as in secure multiparty computation~\cite{wigderson1988completeness,chaum1988multiparty,cramer2000general}, controlling access to servers~\cite{naor1996access}, and generalized oblivious transfer~\cite{shankar2008alternative,tassa2011generalized} (see e.g.,~\cite{beimel2011secret} for further applications).

More recently, different extensions of secret sharing, such as homomorphic secret sharing~\cite{boyle2017foundations} or function secret sharing~\cite{boyle2015function,boyle2016function}, allow the parties not only to share secrets, but also use them to perform computations in a secure way. In this work, we  focus on function secret sharing.  In this setting, one wishes to share a function $f$ and allow the parties to compute $f(x)$ for a chosen input $x$ in a distributed and secret way.
In this case, the $i$-th party is given a share $f_i$ and a common input $x$, from which they can compute the value $y_i$. The recovering algorithm is then able to reconstruct $f(x)$ from $(y_1, \cdots, y_p)$. In order to have a meaningful definition, we require the complexity of the reconstruction algorithm to be much smaller than the complexity of computing $f$. In the classical setting, the recovering algorithm is typically additive~\cite{boyle2015function,boyle2016function}, i.e., $f(x) = \sum_i y_i$.
Function secret sharing  has also found various practical applications such as in multi-server private information retrieval~\cite{chor1997computationally,chor1998private,kushilevitz1997replication} and secure computation with pre-processing~\cite{boyle2019secure}.

The notion of secret sharing has also been extended to the quantum setting, which is similarly defined except with the secret being a quantum state rather than a bitstring~\cite{hillery1999quantum,cleve1999share,karlsson1999quantum}.
This idea has been studied extensively in the literature and has found applications in quantum cryptography such as quantum multi-party computation~\cite{gottesman2000theory,smith2000quantum,guo2003quantum,xiao2004efficient}.

The first contribution of our work is to define {\em quantum function secret sharing} (QFSS). In this primitive, the dealer is a classical computer and each of the parties is in possession of a quantum computer, and the communication between the dealer and the parties is completely classical.  

In our setting, we are interested in computing the probability of running a quantum circuit $C$ on a chosen input $\ket{\psi}$, and projecting it onto a chosen projector $\Pi$. The dealer computes the classical shares $C_1,...,C_p$ from $C$, and sends them to the parties. 

The parties use a common input $\ket{\psi}$, which can either be classically described by the dealer or more generally an unknown quantum state that comes from a quantum process, and they also share a target projection $\Pi$. Given $C_i$, copies of $\ket{\psi}$ and $\Pi$, the $i$-th party can compute a value $y_i$, which is sent back to the dealer (or more generally a third party). The properties that we require from QFSS are the following (see Fig.~\ref{fig:schematic}):
\begin{enumerate}
    \item \textbf{Correctness:} There is an efficient classical reconstruction algorithm that maps $y_1,...,y_p$ to a value $v$ such that $v$ is close to  $||\Pi C\ket{\psi}||^2$.
    \item \textbf{Security (informal): }  Given $C_i$, the $i$-th party has negligible information about $C$.
\end{enumerate}

\begin{figure}
\center
\includegraphics[scale=0.5]
{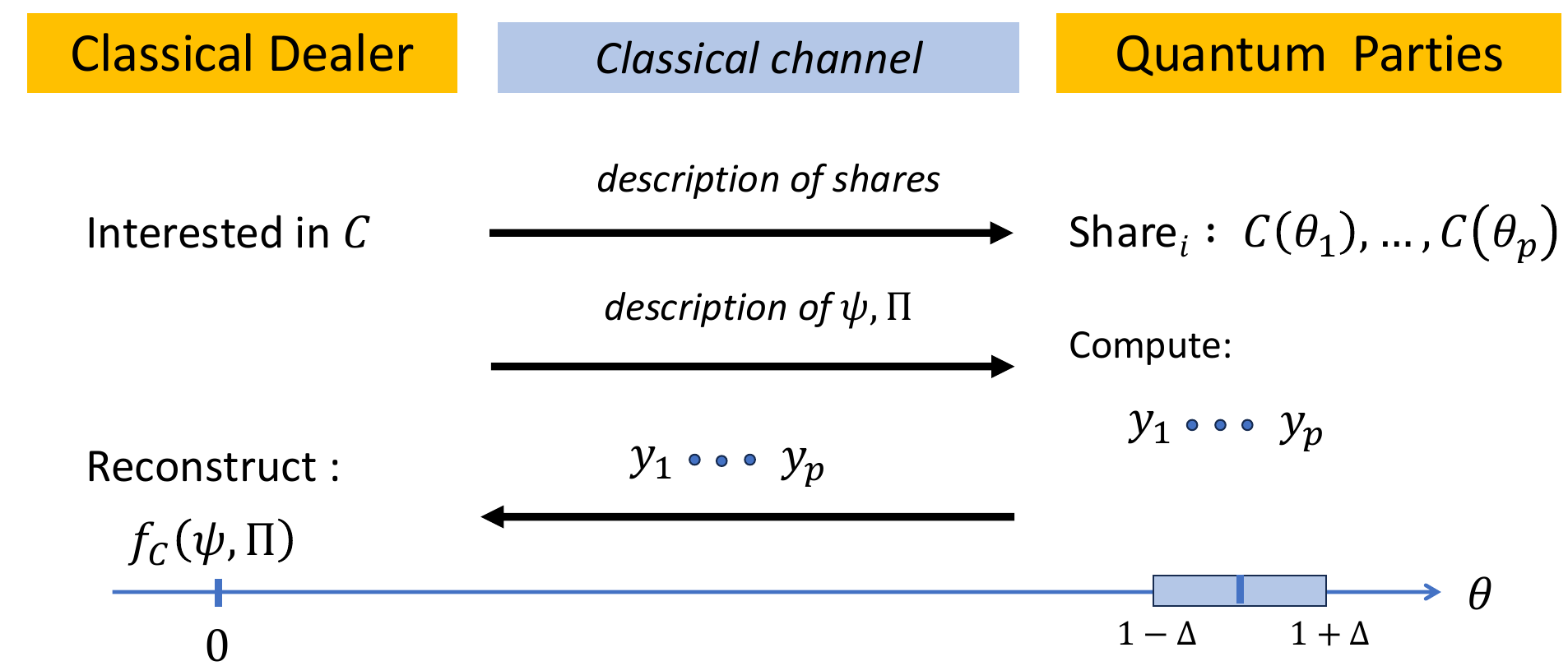}
\caption{A function $f_C$ is computed in a distributed way, where $f_C(\psi,\Pi):=\norm{\Pi C\ket{\psi}}^2$  is the output of a quantum circuit that the dealer (e.g., client) wants to compute. The dealer may wish to provide a classical description of $\psi$ but the protocol allows for a pre-agreed state $\psi$ whose $p$ copies are used by the quantum parties as inputs. The latter allows for inputs whose preparation may require deep quantum circuits. The quantum parties compute $y_i\approx\norm{\Pi C(\theta_i)\ket{\psi}}^2$ and classically communicate them back to the dealer. The reconstruction by the dealer is performed classically and efficiently through an extrapolation.}
\label{fig:schematic}
\end{figure}

\myparagraph{Implementing QFSS with Cayley paths.} 
With our definition in hand, we propose a protocol using the Cayley parametrization of quantum circuits~\cite{movassagh2023hardness,BoulandFLL21}. This technique gives a way of interpolating between quantum circuits.~\footnote{For a more detailed introduction to the topic, see \Cref{sec:cayley}} In particular one can use the Cayley parametrization to interpolate between a fixed circuit and a completely random circuit with the same architecture: Given a real parameter $0 \leq \theta \leq 1$, one defines a  family of circuits $C(\theta)$ such that
\begin{enumerate}
    \item $C(0)$ is the circuit of interest $C$ 
    \item $C(1)$ is a circuit with the same architecture as $C$ whose gates instantiate unitaries drawn independently from the Haar measure. Moreover, it can be shown that the statistical distance of $C(1-\delta)$ and circuits composed of Haar random gates is $\Theta(\Delta)$, for small $|\delta |\le \Delta$.
\end{enumerate}
 A nice  feature of the Cayley parametrization is that it tolerates some noise in the reconstruction of the output probability of interest. For example, given $(\theta_1,y_1),...,(\theta_\ell, y_\ell)$ such that $\left|y_i - \norm{\Pi C(\theta_i)\ket{\psi}}^2\right| \leq \epsilon$, we can compute a rational function $f$ such that \begin{align}\label{eq:reconstruction}
    \left|f(0) - \norm{\Pi C\ket{\psi}}^2\right| \leq \eta,
\end{align}
where $\eta$ depends on the degree of the extrapolation, $\epsilon$, and the extrapolation point $\theta=0$. \

This technique was used for worst-to-average case hardness reductions of computing the output probabilities in the context of quantum supremacy~\cite{movassagh2023hardness}. In this work, we use it to construct a quantum function secret sharing scheme. To the best of our knowledge, this is the first time that such a technique is used in the cryptographic setting.

More concretely, in our scheme we assume that the number of parties $p=\Theta(m)$ is at least the number of the gates $m$ in the circuit.
Each party receives one share $C_i := C(\theta_i)$, where $\Delta \geq |1-\theta_i|$. Each party then runs $C_i$ on the input $\ket{\psi}$ and measures according to $\{\Pi, I- \Pi\}$ to obtain an estimated value $y_i$  by running the circuit  $O(\epsilon^{-2})$ times such that with high probability $\left|y_i -  \norm{\Pi C_i\ket{\psi}}^2\right| \leq \epsilon$. 
Since $f(\theta)$ in \Cref{eq:reconstruction} is a low-degree algebraic function, we can compute $f(0)$ from sufficiently many distinct $(\theta_i,y_i)$, where by low-degree we mean that $\poly(n)$ number of points suffices. However, in order to find an interpolation that is good enough to extrapolate the value $\norm{\Pi C \ket{\psi}}^2$, each party has to provide an approximation of $\norm{\Pi C_i \ket{\psi}}^2$ with exponential precision, requiring exponential-time evaluation algorithms.

Regarding security, given the statistical distance of $C_i$ from the Haar measure, we have that any single party, even with unbounded computational power, cannot distinguish $C_i$ from a circuit composed of Haar random gates except with probability $\Theta(\Delta)$. Therefore, by picking $\Delta$ to be exponentially small in the security parameter, we can show the security of the QFSS scheme. We notice that in the usual quantum protocols, such as QKD and private verification quantum money, the information-theoretical security stems from the fact that the malicious parties actions collapse a quantum state. In our scheme, the information-theoretical security comes from the {\em randomization} properties of quantum circuits offered by the low-degree algebraic Cayley path.

\begin{remark}
In the above we considered a pure state input and output. However our protocol easily permits a generalization to $y_i:=\tr(\rho_i \Pi)$ where the input state is a general density matrix $\Psi$, and the output state of the ith party is $\rho_i=C(\theta_i)\,\Psi\, C(\theta_i)^\dagger$.
\end{remark}

\myparagraph{Limitations of the techniques}
Finally, we prove two limitations for our protocol that are actually inherent to our extrapolation techniques.

First, in \Cref{sec:no-go} we show that in order to achieve correctness and security, we must have an exponentially-precise in the number of gate approximation of $y_i$.

Secondly, we show in \Cref{sec:one party} that the protocol is {\em not} secure against colluding parties by proving that having $C_i$ and $C_j$, for $i \ne j$, suffices to  break the security of our scheme with probability $1$ and in an information-theoretic sense.

\subsection{Discussion and open problems}
In this work, we define the notion of QFSS, propose a protocol for it and prove the limitations of extending our protocol to a broader set of parameters.
This work, to the best of our knowledge, is the first to apply interpolation techniques in a quantum cryptographic setting.
While we focused on the QFSS primitive, we notice that our protocol implements blind delegation of quantum computation, and we hope that it serves as a stepping stone towards the fully classical verification of quantum computation, which was our original aim for this work.

\subsubsection{``Quantum-inspired" FSS}
We recall that in classical function secrete sharing (CFSS)~\cite{boyle2015function} a dealer distributes the secret shares $f_1,...,f_p$ of a desired circuit $f \in \mathcal{F}$,  which is a classical circuit (made up of ANDs and NOTs etc) to the parties. Each party computes their circuit functions by running the circuit on a chosen input $x$ to get $y_i$ and gives back $y_i$ to the dealer who can then calculate $f(x)$. While in standard CFSS schemes, the reconstruction algorithm is additive, one can consider more general algorithms, in which the parties are not strong enough to compute $f$ by themselves (which would trivialize the definition).

Based on our protocol, we can propose a new family of CFSS schemes for functions that can be implemented by {\em quantum circuits} applied on classically described states. The generation and recovery algorithms are exactly the same as ours and the only difference is in the evaluation algorithm: each party uses path-integral ideas to compute the output probabilities of the the circuit provided by the dealer. 

We notice that, currently, we only know general CFSS schemes for general family of functions that can be implemented by classical circuits under (very strong) computational assumptions~\cite{boyle2015function}. With our work, we can achieve an information-theoretically secure scheme for general quantum circuits, with the caveats being the exponential evaluation algorithm and the fact that security only holds against non-colluding adversaries. We compare these two approaches in the following table with the \textcolor{blue}{advantages} and \textcolor{red}{disadvantages}.

~\\

\begin{tabular}{|c|c|c|}
\hline 
 & General FSS from~\cite{boyle2015function} & Quantum(-inspired) FSS (this work)\tabularnewline
\hline 
\hline 
Type of security & \textcolor{red}{Computational} & \textcolor{blue}{Information-theoretic} \tabularnewline
\hline 
Circuits being shared & \textcolor{red}{Classical} & \textcolor{blue}{Quantum} \tabularnewline
\hline 
Evaluation time & \textcolor{blue}{Polynomial} & \textcolor{red}{Exponential}  \tabularnewline
\hline 
$k-$party collusion & \textcolor{blue}{Arbitrary} & \textcolor{red}{$k=1$} (\Cref{cor:impossibility-2qfss}) \tabularnewline
\hline 
\end{tabular}\\

What can one envisage as the utility for this quantum-inspired scheme? One can envisage a scenario in which each party is a classical simulator with a lot of computational power. For example,  super-clusters like Summit or companies such as NVIDIA, Google, IBM, etc. possess such computational capabilities. The dealer, being a small classical computer (e.g., in possession of a laptop), can delegate the quantum computation of interest to the parties. The parties then run classical algorithms such as Tensor Networks to provide exponentially precise answers to the dealer.  Then our work ensures a secure delegation of classical simulation of the quantum circuit of interest and ensure information-theoretic security.

\subsubsection{Private function evaluation}

A similar task to function secret sharing is that of private function evaluation (FPE) ~\cite{abadi1990secure,mohassel2014actively}, which was extended to the quantum setting in~\cite{cao2023quantum}. However, quantum PFE, one wishes to evaluate a quantum function $\mathcal{E}$ on a quantum input shared across multiple parties. Thus, in this setting, the input is shared. In contrast, in QFSS, each party has copies of the input state and the function itself is shared.

\subsubsection{Open problems}
We conclude this section with some open problems.
\medskip

\myparagraph{Alternative techniques for QFSS} In this work, we focused on using the Cayley path for QFSS. We leave it as an open question to find new techniques, extrapolation-based or not, that might improve on ours.
Related open problem is to show that our no-go result extends to all extrapolation-based techniques?

\myparagraph{Computational QFSS}
In this work, we focused on information-theoretical QFSS. A potential direction to overcome the limitations of our technique, especially the limitation on the threshold of our QFSS scheme, is to consider computational QFSS. Of course, one can trivially do it with strong cryptographic primitives such as quantum fully homomorphic encryption (FHE)~\cite{Mahadev18a,Brakerski18}. We leave as an open question if weaker cryptographic assumptions would enable QFSS (potentially for subclasses of quantum circuits).

\myparagraph{Verifiable QFSS} In standard FSS, the parties receive honest shares from the dealer. There is an extension to it, called verifiable FSS~\cite{boyle2016function}, where the parties are able verify if they receive valid shares. More concretely, they receive their shares from a potentially malicious user, and they want the guarantee that the shares are consistent with some function from the valid function space. This can be important for some applications of FSS. We leave it as an open question if verifiable QFSS is possible.

\myparagraph{Verifiable Universal Quantum Blind Computation} In our scheme, the parties are expected to run the honest evaluation so that the reconstruction can compute the desired probability. We ask if our scheme can be extended to check if the parties performed the correct operation, which would allow us to achieve Verifiable Blind Universal Quantum Computation with unentangled servers.

\myparagraph{Concrete application for quantum-inspired FSS} We discussed how the scheme could be of utility. It would be very interesting if a concrete application of practical interest could be found.

\section*{Acknowledgement}
We thank Laura Lewis for discussions. We also thank Kewen Wu, Ryuhei Mori, Jarrod McClean, and Stephen Jordan. ABG is supported by ANR JCJC TCS-NISQ ANR-22-CE47-0004, and by the French National Research Agency award number ANR-22-PNCQ-0002. This work was done in part while the authors were visiting the Simons Institute for the Theory of Computing, supported by DOE QSA grant $\#$FP00010905. ABG thanks Google Quantum AI for a collaborative grant.

\section{Preliminaries}

\subsection{Notation}

Fix the following notation. For any tuple $x = (x_1,\dots, x_m)$ and set $A = \{i_1,\dots, i_{|A|}\} \subseteq [m]$, then $x_A$ denotes the tuple $(x_{i_1},\dots, x_{i_{|A|}})$.
Similarly, for domains $X_1,\dots, X_m$, denote $X_A$ as $X_{i_1}\times \cdots \times X_{i_{|A|}}$.
In this way, if $x = (x_1,\dots, x_m) \in X_1\times \cdots \times X_m$, then $x_A \in X_A$.

\subsection{Cayley path}\label{sec:cayley}
The Cayley path is a parameterized path that interpolates between two unitaries~\cite{movassagh2023hardness}.
Namely, let $\mathbb{U}(N)$ be the set of $N \times N$ unitary matrices and $U_0, U_1 \in \mathbb{U}(N)$.
Then, the Cayley path is a parameterized path $U(\theta)$ such that $U(\theta) \in \mathbb{U}(N)$ for all $\theta \in [0,1]$ and $U(0) = U_0$ and $U(1) = U_1$ based on the Cayley function
\begin{equation}
f(x) = \frac{1 + ix}{1 - ix},
\end{equation}
where we define $f(-\infty) = -1$.
Here, we (informally) discuss the Cayley path for interpolating between a quantum circuit $C$ and a Haar random circuit following the presentation in~\cite{movassagh2023hardness}.

Let $C = \mathcal{C}_m\cdots \mathcal{C}_1$ be a quantum circuit acting on $n$ qubits with $m$ local unitary gates.
Suppose that $\mathcal{C}_k = C_k \otimes I$ so that $C_k$ is a one or two-qubit gate with identity on the rest of the qubits.
Given a local gate $H_k$, then there exists a unique Hermitian matrix $h_k$ such that $H_k = f(h_k)$.
Suppose that $H_k$ is Haar random gate that is the same size as a gate $C_k$ in the quantum circuit.
Then, the Cayley path for each gate is given by
\begin{equation}\label{eq:c(theta)}
C_k(\theta) = C_k f(\theta h_k),
\end{equation}
where $H_k = f(h_k)$, thus $C_k(0) = C_k$ and $C_k(1) = C_kH_k$ is a Haar random gate.
The full quantum circuit has the Cayley path $C(\theta) = \mathcal{C}_m(\theta) \cdots \mathcal{C}_1(\theta)$, where $\mathcal{C}_k(\theta) = C_k(\theta) \otimes I$.

Crucially,~\cite{movassagh2023hardness} showed that if one runs $C(\theta_i)$ for different choices of $\theta_i$ and estimates a measurement $p(\theta_i)$ probability for a certain bitstring, then one can recover a function $p(\theta)$ that gives the measurement probability for $C(\theta)$ for any choice of $\theta$ as long as the algorithm is given enough points with the necessary precision.
Recall from the construction of the Cayley path~\cite{movassagh2023hardness} that the probabilities are rational function of type $(\Theta(m),\Theta(m))$ with factorizable denominators whose determination can be reduced to polynomial extrapolation~\cite{movassagh2023hardness}. The recovery can be done via Lagrange extrapolation as shown in Eq.~13 in~\cite{kondo2022quantum}.  Below we denote by  $e(\theta)$ the error probability which is the difference of the exact  and estimated probabilities obtained from Lagrange extrapolation or generalized Berlekamp-Welch~\cite{BoulandFLL21}
\begin{proposition}[Lemma 4 via Eq.~13 in~\cite{kondo2022quantum}]
\label{thm:general-bw}
Suppose $e(\theta)$ is a degree $d$ polynomial in $\theta$. Assume $e(\theta_i)\le \epsilon$ where $|1-\theta_i|\in [0,\Delta]$. Then   $|e(0)| \leq \epsilon \frac{\exp[d(1+\log\Delta^{-1})]}{\sqrt{2\pi d}}$.
\end{proposition}
Let $p_{C(1)}$ be the probability distribution over circuits with Haar local unitary gates and $p_{C(\theta)}$ the probability distribution over circuits whose gates are deformed according to Eq.~\eqref{eq:c(theta)} then
\begin{proposition}[Lemma 2 in~\cite{movassagh2023hardness}] \label{thm:statistical-distance} The total variation distance $\text{TVD}\left(p_{C(\theta)},p_{C(1)}\right)=O(m\Delta)$, for $|1-\theta| \leq \Delta$.
\end{proposition}

\section{Quantum Function Secret Sharing}
In \Cref{sec:definition-qfss}, we define the concept of quantum function secret sharing, and then we follow with our scheme in \Cref{sec:qfss-scheme}.

\subsection{Definition}
\label{sec:definition-qfss}

Let $C$ be an $n$-qubit quantum circuit with $m$ gates, $\Pi$ be an $n$-qubit projector and $\ket{\psi}$ be an $n$-qubit quantum state. The functions we wish to share are of the form
\begin{equation}
f_C(\ket{\psi},\Pi) = ||\Pi C\ket{\psi}||^2
\end{equation}
for circuits $C \in \mathcal{C}$.
In other words, $f \in \mathcal{F}$ is the probability of projecting $C\ket{\psi}$ onto $\Pi$.

\label{sec:def}
\begin{definition}\label{def:qfss}[$(p, t, \ell, T, \eta)$ decisional quantum function secret sharing]
Let $p \in \mathbb{N}$ be the number of parties, $\mathcal{C}$ be a family of circuits and $\eta \in \mathcal{R}$ be the desired precision. Let $S_1,\dots, S_p$ be the share spaces for each of the $p$ parties.
A \emph{$t$-secure $p$-party quantum function sharing scheme} is a tuple of functions $(\Gen, \Eval, \Rec)$, where $\Gen$ is called the generation algorithm, $\Eval$ is called the evaluation algorithm, and $\Rec$ is called the reconstruction algorithm. These satisfy the following syntax
\begin{itemize}
\item $\Gen(1^\lambda, C)$: On input the security parameter $1^\lambda$ and a description of a circuit $C$, the generation algorithm outputs $p$ shares $(C_1,...,C_p) \in S_1 \times ... \times S_p$. %
\item $\Eval(i, C_i, \ket{\psi}^{\otimes \ell}, \Pi)$: On input the party index $i \in [p]$, $\ell$ copies of an input $\ket{\psi}$ and the description of a projector $\Pi$, the evaluation algorithm outputs  $y_i$ in time T. 
\item $\Rec(%
y_1,\dots, y_p)$: On inputs $y_i \in S_i$, the reconstruction algorithm computes $v \in [0,1]$ in time $\poly(\lambda,p)$.
\end{itemize}
Moreover, $\Gen, \Eval,$ and $\Rec$ satisfy the following properties:
\medskip

\noindent\textbf{Correctness}: The secret function can be recovered from all of the shares, i.e., for all $C$, $\ket{\psi}$ and $\Pi$, if $(C_1,\dots, C_p) \leftarrow \Gen(1^\lambda, C)$, then
\begin{align*}
\Pr[\left|\Rec(%
\Eval(1,C_1,\ket{\psi}^{\otimes \ell}, \Pi),\dots, \Eval(p, C_p, \ket{\psi}^{\otimes \ell},\Pi)) -  ||\Pi C\ket{\psi}||^2 \right| \leq \eta] \geq 1 - \negl(\lambda).
\end{align*}

\noindent\textbf{Security}:
Let us consider the $\mathsf{IndQFSS}$ security game, which is parametrized by some $T \subseteq [p]$, described in Algorithm~\ref{algo:ind-qfss} below.
\begin{algorithm}
\label{algo:ind-qfss}
\caption{Indistinguishability game for function secret sharing $\mathsf{IndQFSS}_T$}
\KwIn{Function secret sharing scheme $(\Gen, \Eval, \Rec)$, corrupted parties $T \subseteq [p]$}
The adversary $\mathcal{A}$ outputs a pair of circuits $C^0, C^1 \in \mathcal{C}$.\\
The challenger samples $b \leftarrow \{0,1\}$ uniformly at random and computes $(C^b_1,\dots, C^b_p)
\leftarrow \Gen(1^\lambda,C^b)$.\\
The challenger sends $(C^b_i)_{i \in T}$ to $\mathcal{A}$.\\
The adversary $\mathcal{A}$ outputs a guess $b' \in \{0,1\}$.\\
Output $b \oplus b'$.
\end{algorithm}

The protocol is statistically secure if for every $T$ such that $|T|< t$ and every adversary $\mathcal{A}$,
\begin{align}\label{eq:security}
\left|\Pr(\mathsf{IndQFSS}_T(\mathcal{A}) = 0) - \frac{1}{2}\right| = \negl(\lambda).
\end{align}

The protocol is computationally secure if \Cref{eq:security} holds for every QPT adversary $\mathcal{A}$.
\end{definition}

\subsection{Quantum secret sharing from Cayley paths}
\label{sec:qfss-scheme}
Our function sharing scheme is defined by the algorithms $(\Gen, \Eval, \Rec)$ as given Algorithms~\ref{algo:gen},~\ref{algo:eval},~\ref{algo:rec}.
The intuition is that to generate the keys corresponding to shares of the function (Algorithm~\ref{algo:gen}), we garble the circuit we wish to share for different choices of the parameter $\theta$ in the Cayley path. The evaluation (Algorithm~\ref{algo:eval}) is just running the randomized circuit on a given input state many times to estimate the output probability. Finally, reconstruction (Algorithm~\ref{algo:rec}) works by using techniques from~\cite{movassagh2023hardness} such as the generalized BW algorithm (\Cref{thm:general-bw}) or Lagrange interpolation.

\begin{algorithm}
\label{algo:gen}
\caption{Generation algorithm $\Gen$}
\KwIn{Security parameter $\lambda$, description of $n$-qubit $m$-gate circuit $C \in \mathcal{C}$, number of parties $p \geq m$}
\KwOut{Circuits $(C_1,\dots, C_p)$}
Construct the Cayley path $C(\theta)$.\\
Compute $C_i = C(1-i \Delta(\lambda))$, for some negligible function $\Delta$. \\ 
Output $(C_1,...,C_p)$.
\end{algorithm}

\begin{algorithm}
\label{algo:eval}
\caption{Evaluation algorithm $\Eval$}
\KwIn{Index $i \in [p]$, circuit $C_i$, $\ell$ copies of $n$-qubit input quantum state $\ket{\psi}$, and the description of a projector $\Pi$.}
\KwOut{$y_i \in [0,1]$.}
\For{$j \in \{1,...,\ell\}$}{Run $C_i$ on a copy of $\ket{\psi}$ and measure according to $\{\Pi, I - \Pi\}$ with outcome $b_j$.}
Output $y_i = \frac{1}{\ell}\sum_{j = 1}^i b_j$.
\end{algorithm}

\begin{algorithm}
\label{algo:rec}
\caption{Reconstruction algorithm $\Rec$}
\KwIn{$(y_1,\dots, y_p) \in [0,1]^p$}
\KwOut{$y$}
Use $(1- i\Delta(\lambda), y_i)$ as samples for Lagrange interpolation to obtain $F(\theta)$. \\
Output $y = F(0)$.
\end{algorithm}

We will show the following theorem.

\begin{theorem}\label{thm:one-secret-sharing}
    Let $\mathcal{C}$ be the a family of circuits on $n$ qubits and $m$ gates where the architecture is fixed but the the gates could take an arbitrary value. For any $p \geq m$ and for security parameter $\lambda$,
    $(\Gen, \Eval, \Rec)$, defined in Algorithms~\ref{algo:gen},~\ref{algo:eval},~\ref{algo:rec}, is a   
     $\left(p, 1,  \ell, |C| \ell, \eta\right)$ quantum function secret sharing, where $\ell = \omega(\log\lambda) \left(\left(\frac{\exp[d(1+\log\Delta^{-1}]}{\eta\sqrt{2\pi d}}\right)^2\right)$.
\end{theorem}

We will split the proof of \Cref{thm:one-secret-sharing} in two parts. In \Cref{lem:correctness}, we prove the correctness of the scheme and in \Cref{lem:security}, we prove its security.

\begin{lemma}\label{lem:correctness}
For all $C$, $\ket{\psi}$ and $\Pi$, if $(C_1,\dots, C_p) \leftarrow \Gen(1^\lambda, C)$, then
\begin{equation}
\Pr[\left|\Rec(\Eval(1,C_1,\ket{\psi}^{\otimes \ell},\Pi),\dots, \Eval(p, C_p, \ket{\psi}^{\otimes \ell},\Pi)) -  ||\Pi C\ket{\psi}||^2 \right| \leq \eta] \geq 1 - \negl(\lambda).
\end{equation}   
\end{lemma}
\begin{proof}
Let $y_i$ be the output of the $\Eval$ procedure on $(i,C_i,\ket{\psi}^{\otimes \ell},\Pi)$. We have that by Hoeffding's bound, 
\begin{align}\label{eq:bound-approximation}
\Pr[|y_i - ||\Pi C_i \ket{\psi}||^2| \geq  \epsilon] \leq 2e^{-\ell\epsilon^2}.
\end{align}

In this case, by the union bound, with probability at least $1 - p 2e^{-\ell\epsilon^2}$, for all $i \in \{1,...,p\}$, we have that
\begin{align}
    |y_i - ||\Pi C_i \ket{\psi}||^2| \leq  \epsilon.
\end{align} 

Using that $\ell = \omega(\log\lambda) \left(\left(\frac{\exp[d(1+\log\Delta^{-1}]}{\eta\sqrt{2\pi d}}\right)^2\right)$ and setting $\epsilon = \frac{\eta\sqrt{2\pi d}}{\exp[d(1+\log\Delta^{-1})]}$, it follows by \Cref{thm:general-bw} that with probability $1 - \negl(\lambda)$, Lagrange used in $\Rec$ outputs a value $y$ such that
\begin{align*}
    |y - ||\Pi C\ket{\psi}||^2| \leq \eta.
\end{align*}
where $\eta=\epsilon \frac{\exp[d(1+\log\Delta^{-1})]}{\sqrt{2\pi d}}.$
\end{proof}

\begin{lemma}\label{lem:security}
 For any $i$ and any QPT adversary $\mathcal{A}$,
\begin{equation}
\label{eq:security-def}
\Adv_{\mathsf{IndQFSS}_{\{i\}}}(\mathcal{A}) \leq \negl(\lambda),
\end{equation}
where $\mathsf{IndQFSS}$ is the game described in Algorithm~\ref{algo:ind-qfss}.
\end{lemma}
\begin{proof}
    Let us consider two circuits $C^0$ and $C^1$ from the same architecture, and their respective Cayley paths $C^0(\theta)$ and $C^1(\theta)$, constructed as described in \Cref{sec:cayley}.

    We notice that from \Cref{thm:statistical-distance}, the statistical distance between $C^b(1)$ and $C^b(\theta)$ is at most $p\Delta(\lambda)$ for all $\theta \geq 1 - p\Delta(\lambda)$. Moreover, the statistical distance between $C^0(1)$ and $C^1(1)$ is $0$, since both of them are constructed by picking Haar random gates. Using the triangle inequality, the statistical distance between $C^0(\theta)$ and $C^1(\theta)$ is at most $2p\Delta(\lambda)$ for all $\theta \geq 1 - p\Delta(\lambda)$. Since $\Delta = \negl(\lambda)$, the result follows.
\end{proof}

\section{Super-polynomial precision is required}\label{sec:no-go}

The goal of this section is to prove that exponential precision is required for negligible security in our scheme. 
In order to prove this, we first show in \Cref{sec:lower-bound-information-extracted} that our scheme achieves QFSS with negligible security only if we pick $\Delta = \negl(\lambda)$. Then, we show in \Cref{sec:negl-Delta} that if $\Delta = \negl(\lambda)$, then exponential precision is required for the reconstruction.

\subsection{Lower bound on the information extracted}
\label{sec:lower-bound-information-extracted}
 We prove a lower-bound on total variation distance (TVD) between close to Haar randomized circuits that provides a distinguishability result and hence a no-go theorem (Theorem~\ref{thm:no-go} below). We first establish a lower-bound on the TVD between two $\Delta-$close to Haar randomized gates of finite size and then in Theorem~\ref{thm:no-go} establish a lower-bound on the TVD for two $\Delta-$close to Haar randomized circuits.

The Haar measure for unitary matrices is the unique invariant measure
$d\mu(U)$ that is uniform for all $U\in\mathbb{U}(N)$ where $\mathbb{U}(N)$
is the unitary group of size $N$. Recall the Cayley parametrization
\[
f(\theta h)=\frac{1+i\,\theta h}{1-i\,\theta h}\;,\qquad h=h^{\dagger}.
\]
We choose $h$ and $h'$ such that $f(h)$ and $f(h')$ are independent instances of Haar
measure unitaries. We wish to lower-bound the TVD
between $G_{1}f(\theta h)$ and $G_{2}f(\theta h')$ for $|1-\theta|\le\Delta\ll1$
where at $\theta=1$ both $G_{1}f(h)$ and $G_{2}f(h')$ are Haar distributed
by the translation invariance of the Haar measure. By unitary invariance
of norms it is sufficient to find the TVD between $f(\theta h)$ and
$Gf(\theta h')$ where $G=G_{1}^{\dagger}G_{2}$. 

Let us define two random variables $X(\theta):=f(\theta h)$ and $\tilde{X}(\theta)=Gf(\theta h')$
whose respective probability densities we denote by $p_{X(\theta)}$
and $p_{\tilde{X}(\theta)}$. Since $\text{Pr}[Gf(\theta h')=U] = \text{Pr}[Gf(\theta h)=U]$, it is easy to see that 
\[
p_{\tilde{X}(\theta)}(U) = \text{Pr}[GX(\theta)=U]=\text{Pr}[X(\theta)=G^\dagger U]
\]
and we conclude that for all $\theta$
\[
p_{\tilde{X}(\theta)}(U)=p_{X(\theta)}(G^{\dagger}U).
\]
The TVD between the two distributions is defined by
\begin{eqnarray}
\text{TVD}(p_{\tilde{X}(\theta)}\,,p_{X(\theta)}) & := & \frac{1}{2}\int_{\mathbb{U}}\left|p_{\tilde{X}(\theta)}(U)-p_{X(\theta)}(U)\right|\;dU\nonumber \\
 & = & \frac{1}{2}\int_{\mathbb{U}}\left|p_{X(\theta)}(G^{\dagger}U)-p_{X(\theta)}(U)\right|\;dU\label{eq:TVD}
\end{eqnarray}

\subsubsection{Haar measure and induced measures under Cayley transformation}

Here we derive the probability density $p_{X(\theta)}$. Let $U=\frac{1+ih}{1-ih}$
be the Cayley transform with $h=h^{\dagger}$, the Haar measure $d\mu(U)$
induces a measure over $h=f^{-1}(U)=i\frac{1-U}{1+U}$ that after
ignoring an unimportant constant pre-factor writes~\cite{toyama1948haar} 
\begin{eqnarray*}
dU\; & = & \frac{dh}{\left[\det\left(I+h^{2}\right)\right]^{N}}\;
\end{eqnarray*}
We are interested in the induced distribution over unitaries under
the transformation $Y:=\theta h$. We write $U(y)=f(h(y))$ where
$h(y)=y/\theta$. The matrices form a vector space in $N^{2}$ dimensional
space. The composite map has a determinant that is the product of
the determinants of each map. By the chain rule of matrix Jacobians
we have $dU=\left|f'\left(h(y)\right)h'(y)\right|dy$. Since $h'(y)=1/\theta$,
using Toyama's result we have $dU=p_{Y}(y)dy$ where
\begin{eqnarray}
p_{Y}(y)dy\; & = & \theta^{-N^{2}}\frac{dy}{\left[\det\left(1+(y/\theta)^{2}\right)\right]^{N}}.\nonumber \\
 & = & \theta^{N^{2}}\frac{dy}{\left[\det\left(\theta^{2}+y^{2}\right)\right]^{N}}\label{eq:p_Y}
\end{eqnarray}
To make use of Eq.~\eqref{eq:TVD} we invoke the inverse function
theorem to obtain the density over the (no longer uniform) unitaries
$x:=X(\theta)=f(y)$. Since $y=f^{-1}(x)=i(I-x)(I+x)^{-1}$ we have
$dy=id[1-x]\,(I+x)^{-1}+i(I-x)d[(1+x)^{-1}]$. By an application
of chain rule on $dI=d[(1+x)^{-1}(1+x)]=0$ we obtain $d[(1+x)^{-1}]=-(1+x)^{-1}dx(1+x)^{-1}$. Hence we arrive at
\begin{eqnarray*}
dy & = & -2i\,(1+x)^{-1}dx(1+x)^{-1}
\end{eqnarray*}
We now make a substitution
in Eq.~\eqref{eq:p_Y} and write
\begin{eqnarray*}
p_{X(\theta)}(x)dx\; & = & \theta^{N^{2}}\frac{|dy[{x,dx}]|}{\left[\det\left(\theta^{2}-\left(\frac{1-x}{1+x}\right)^{2}\right)\right]^{N}}.
\end{eqnarray*}
We apply the determinant of products to calculate the Jacobian and
write
\begin{eqnarray}
p_{X(\theta)}(x)dx\; & \propto & (2\theta)^{N^{2}}\frac{dx}{\left|\det\left(\theta^{2}(1+x)^{2}-(1-x)^{2}\right)\right|^{N}}\;.\label{eq:pX}
\end{eqnarray}
Recall from Eq.~\eqref{eq:TVD} that $p_{\tilde{X}(\theta)}(x)dx=p_{X(\theta)}(G^{\dagger}x)dx$,
and we have
\begin{eqnarray}
p_{\tilde{X}(\theta)}(x)dx\; & \propto & (2\theta)^{N^{2}}\frac{dx}{\left|\det\left(\theta^{2}(1+G^{\dagger}x)^{2}-(1-G^{\dagger}x)^{2}\right)\right|^{N}}\;.\label{eq:pCX}
\end{eqnarray}

\subsubsection{Total variation distance}

We are now in a position to compute $\text{TVD}(p_{\tilde{X}(\theta)}\,,p_{X(\theta)})$
in Eq.~\eqref{eq:TVD} by using the aforementioned expressions in Eqs.~\eqref{eq:pX} and \eqref{eq:pCX}. We are interested in $|1-\theta|<\Delta=1/\text{poly}(n)$.
We let $\theta=1-\delta$ where $|\delta|\le\Delta\ll1$ and ignoring
terms of $O(\delta^{2})$ we write
\begin{eqnarray*}
\det\left|\left((1-\delta)^{2}(1+x)^{2}-(1-x)^{2}\right)\right|^{-N} & = & |\det(4x)|\left|\det\left(1-\frac{\delta}{2}\left(x^{1/2}+x^{-1/2}\right)^{2}\right)\right|^{-N}\\
 & = & 4^{N}\left[1+N\frac{\delta}{2}\text{Tr}\left[\left(x^{1/2}+x^{-1/2}\right)^{2}\right]\right],
\end{eqnarray*}
where we used $\det(1+\epsilon A)\approx1+\epsilon\text{Tr}(A)$ for
small $\epsilon$. Similarly we derive
\begin{eqnarray*}
\det\left|\left((1-\delta)^{2}(1+G^{\dagger}x)^{2}-(1-G^{\dagger}x)^{2}\right)\right|^{-N} & = & 4^{N}\left[1+N\frac{\delta}{2}\text{Tr}\left[\left((G^{\dagger}x)^{1/2}+(G^{\dagger}x)^{-1/2}\right)^{2}\right]\right].
\end{eqnarray*}
Therefore
\begin{eqnarray*}
p_{X(\theta)}(x)-p_{\tilde{X}(\theta)}(x) & := & \delta\frac{N}{2}4^{N}\text{Tr}\left[\left(x^{1/2}+x^{-1/2}\right)^{2}-\left((G^{\dagger}x)^{1/2}+(G^{\dagger}x)^{-1/2}\right)^{2}\right]\\
 & = & \delta\frac{N}{2}4^{N}\text{Tr}\left[(I-G^{\dagger})x+x^{-1}(I-G)\right]
\end{eqnarray*}
We use these to calculate the TVD
\begin{eqnarray*}
\text{TVD}(p_{\tilde{X}(\theta)}\,,p_{X(\theta)}) & := & \delta\,N4^{N-1}\int_{\mathbb{U}}\left|\text{Tr}\left[(x+x^{\dagger})-(G^{\dagger}x+x^{\dagger}G)\right]\right|\;dx
\end{eqnarray*}
By the invariance of Haar measure under $x\mapsto G^{1/2}x$ we write
\begin{eqnarray*}
\text{TVD}(p_{\tilde{X}(\theta)}\,,p_{X(\theta)}) & := & \delta\frac{N}{2}4^{N}\int_{\mathbb{U}}\left|\text{Tr}\left[(G^{1/2}-G^{-1/2})x-x^{-1}(G^{1/2}-G^{-1/2})\right]\right|dx\\
 & = & \delta N4^{N}\int_{\mathbb{U}}\left|\text{Tr}\left[V_{G}\sin\left(\frac{\Lambda_{G}}{2}\right)V_{G}^{\dagger}\;x-x^{\dagger}\;V_{G}\sin\left(\frac{\Lambda_{G}}{2}\right)V_{G}^{\dagger}\right]\right|dx
\end{eqnarray*}
where $G=V_{G}\exp(i\Lambda_{G})V_{G}^{\dagger}$ and $\Lambda_{G}=\text{diag}(\lambda_{1},\dots,\lambda_{N})$
with $0<\lambda_{j}\le2\pi$. By the cyclic property of trace we have,
and the invariance of Haar measure under right and left multiplication
by unitaries we have under $x\mapsto V_{G}xV_{G}^{\dagger}$, we have
\begin{eqnarray*}
\text{TVD}(p_{\tilde{X}(\theta)}\,,p_{X(\theta)}) & := & \delta N4^{N}\int_{\mathbb{U}}\left|\text{Tr}\left[\sin\left(\frac{\Lambda_{G}}{2}\right)\left(x-x^{\dagger}\right)\right]\right|dx
\end{eqnarray*}
Since $\Lambda_{G}$ is diagonal $\text{Tr}[x\,\sin\left(\frac{\Lambda_{G}}{2}\right)]=\sum_{i=1}^{N}x_{ii}\,\sin(\frac{\lambda_{i}}{2})$
and we have
\begin{eqnarray*}
\text{TVD}(p_{\tilde{X}(\theta)}\,,p_{X(\theta)}) & := & \delta N4^{N}\int_{\mathbb{U}}\left|\sum_{i=1}^{N}\left(x_{ii}-\overline{x_{ii}}\right)\,\sin(\lambda_{i}/2)\right|dx\\
 & = & \delta N2\cdot4^{N}\int_{\mathbb{U}}\left|\sum_{i=1}^{N}\text{Im}(x_{ii})\,\sin(\lambda_{i}/2)\right|dx
\end{eqnarray*}
Because $G$ is fixed, $\sin(\Lambda_{G}/2)$ is fixed. In order to lower-bound this quantity we are free to choose $G_1$ and $G_2$ to induce $O(1)$ eigenvalues  away from  one. We can more generally state that  for (non-trivial) $G\ne I$ that is sufficiently far from identity, the
set of solutions for $\sum_{i=1}^{N}\text{Im}(x_{ii})\sin(\lambda_{i}/2)=0$
is a low dimensional hyperplane. The integral involves a fixed $G\ne I$,
a finite $N\in\{2,4\}$ , and the Haar measure. Since the integrand
does not vanish with probability one (over the full measure set),
the integral must evaluate to a positive constant. This establishes
that $\text{TVD}(p_{\tilde{X}(\theta)}\,,p_{X(\theta)})=\Theta(\delta)\le \Theta(\Delta)$.
QED

\begin{remark}
The above proof holds because of the nice properties of the Cayley parametrization, and such a lower-bound does not hold in general for other parametrizations. We numerically investigated the above by sampling $10^{7}$ unitaries from the Haar measure and find $\mathbb{E}|\text{Tr}(x+x^{\dagger})|\approx2.16$. Further, for a fixed $G$ drawn randomly from Haar unitary measure we find
\[
\mathbb{E}\left[\left|\text{Tr}\left[(x+x^{\dagger})-(G^{\dagger}x+x^{\dagger}G)\right]\right|\right]\approx3.07
\]
which fluctuates with the choice of $G$. 
\end{remark}

Proving $\Omega(\Delta)$ statistical distance between $C(\theta)$ and $B(\theta)$ immediately follows from the results of this section because of the fact that the statistical distance can only increase if we have more gates. Below we prove a better lower-bound. 
\begin{theorem}\label{thm:no-go} The total variation distance between the probability distribution over circuits $C(\theta)$ and $B(\theta)$ is $\Theta(m\Delta)$ where $m$ is the number of gates.
\end{theorem}
\begin{proof}
The upper-bound follows from previous work~\cite{movassagh2023hardness}. The lower-bound on a single gate was derived above in Section~\ref{sec:no-go}, where $G_1$ and $G_2$ are seen as corresponding gates of $C(\theta)$ and $B(\theta)$. By the independence of the gates, the measure over each circuit is product, i.e.,  $p_{X_1(\theta),X_2(\theta),\dots,X_m(\theta)}=\prod_{i=1}^m p_{X_i(\theta)}$ and similarly $p_{\tilde{X}_1(\theta),\tilde{X}_2(\theta),\dots,\tilde{X}_m(\theta)}=\prod_{i=1}^m p_{\tilde{X}_i(\theta)}$. Total variation distance is the maximum possible difference between two distributions hence the minimum overlap is
\[
1-\text{TVD}(p_{X_1(\theta),\dots,X_m(\theta)},p_{\tilde{X}_1(\theta),\dots,\tilde{X}_m(\theta)}) \le  (1-\Theta(\Delta))^m \qquad \Delta=O(m^{-1}). 
\]
We conclude that $\text{TVD}(p_{X_1(\theta),\dots,X_m(\theta)},p_{\tilde{X}_1(\theta),\dots,\tilde{X}_m(\theta)})\ge\Omega(m\Delta).$
\end{proof}

This theorem can be directly used to construct an attack to our scheme if we pick $\Delta = \frac{1}{\poly(\lambda)}$.

\begin{corollary}
 The triple   $(\Gen, \Eval, \Rec)$, defined in Algorithms~\ref{algo:gen},~\ref{algo:eval},~\ref{algo:rec}, is {\em not}  a
     $\left(p, 1,  \ell, |C| \ell, \eta\right)$ QFSS scheme if $\Delta = \frac{1}{\poly(\lambda)}$.
\end{corollary}

\subsection{Exponential precision}
\label{sec:negl-Delta}
This section shows that in our scheme, we need superpolynomial $\ell$ in order to achieve negligible security.
In order to show that, we will prove that if $\ell = \poly(\lambda)$ and $\Delta = \negl(\lambda)$, then there exists a circuit $C^*$ for which $\Rec$ will fail with constant probability.

\begin{lemma}
For any $\ell = \poly(\lambda)$ and $\Delta = \negl(\lambda)$, there exists a circuit $C$ such that
\begin{align}\label{eq:reconstruction-fail}
    \Pr\left[
    \left|\Rec(y_1,...,y_p) - \norm{\Pi C_i \ket{\psi}}^2 \right| =1  \quad \middle| 
    \quad 
    \substack{
(C_1,\dots, C_p) \leftarrow \Gen(1^\lambda, C^*) \\
y_i = \Eval_\ell(i, C^*_i \ket{\psi}, \Pi) 
}
    \right] = 1- \negl(\lambda).
\end{align}
\end{lemma}
\begin{proof}
Let us fix $\ket{\psi} = \ket{0^n}$, $\Pi = \ketbra{0^n}{0^n}$, and $C$ to be a circuit composed of one layer of one-qubit identities. Notice that 
\begin{align}\label{eq:example-expected}
    \norm{\Pi C \ket{\psi}}^2  = 1.
\end{align}

Let $(C_1,\dots, C_p) \leftarrow \Gen(1^\lambda, C)$ and $\mathcal{C}_i$ be the distribution of the circuit $C_i$. Let $\mathcal{H}$ be the distribution where we have one layer of one-qubit Haar random gates.
We have from \Cref{thm:no-go} that $\text{TVD}(\mathcal{C}_i,\mathcal{H}) = O(m\Delta(\lambda))$. 
Since $\mathbb{E}_{D \sim \mathcal{H}}[\norm{\Pi D \ket{\psi}}^2] = \frac{1}{2^n}$ and $\text{Var}_{D \sim \mathcal{H}}[\norm{\Pi D \ket{\psi}}^2] = \frac{1}{2^{2n}}$, we have that for every polynomial $p$,
    \begin{align*}
        \Pr_{C_i \sim \mathcal{C}_i}\left[\norm{\Pi C_i \ket{\psi}}^2  \geq \frac{1}{p(\lambda)}\right] = \negl(\lambda),
    \end{align*}

In this case, except with negligible probability, for all $i \in [p]$, $\Eval(i,C_i,\ket{\psi}^{\otimes \ell},\Pi) = 0$, since $\ell = \poly(\lambda)$. It follows that with probability $1 - \negl(\lambda)$
\begin{align}\label{eq:example-reconstruction}
    \Rec\left( \Eval\left(1,C_1,\ket{\psi}^{\otimes \ell},\Pi\right),\cdots,\Eval\left(p,C_p,\ket{\psi}^{\otimes \ell},\Pi\right) \right) = \Rec(0,\cdots,0) = 0.
\end{align}

\Cref{eq:reconstruction-fail} follows from \Cref{eq:example-expected,eq:example-reconstruction}.
\end{proof}

This shows that we cannot pick $\ell = \poly(\lambda)$ in our scheme since we do not have its correctness.

\begin{corollary}
 The triple   $(\Gen, \Eval, \Rec)$, defined in Algorithms~\ref{algo:gen},~\ref{algo:eval},~\ref{algo:rec}, is {\em not}  a
     $\left(p, 1,  \ell, |C| \ell, \eta\right)$ QFSS scheme for $\ell = \poly(\lambda)$.
\end{corollary}

\section{Impossibility beyond one-party secret sharing}\label{sec:one party}

In this section, we prove that if two malicious parties collude, they can break our QFSS scheme.

\begin{figure}
\center
 \includegraphics[scale=0.5]
{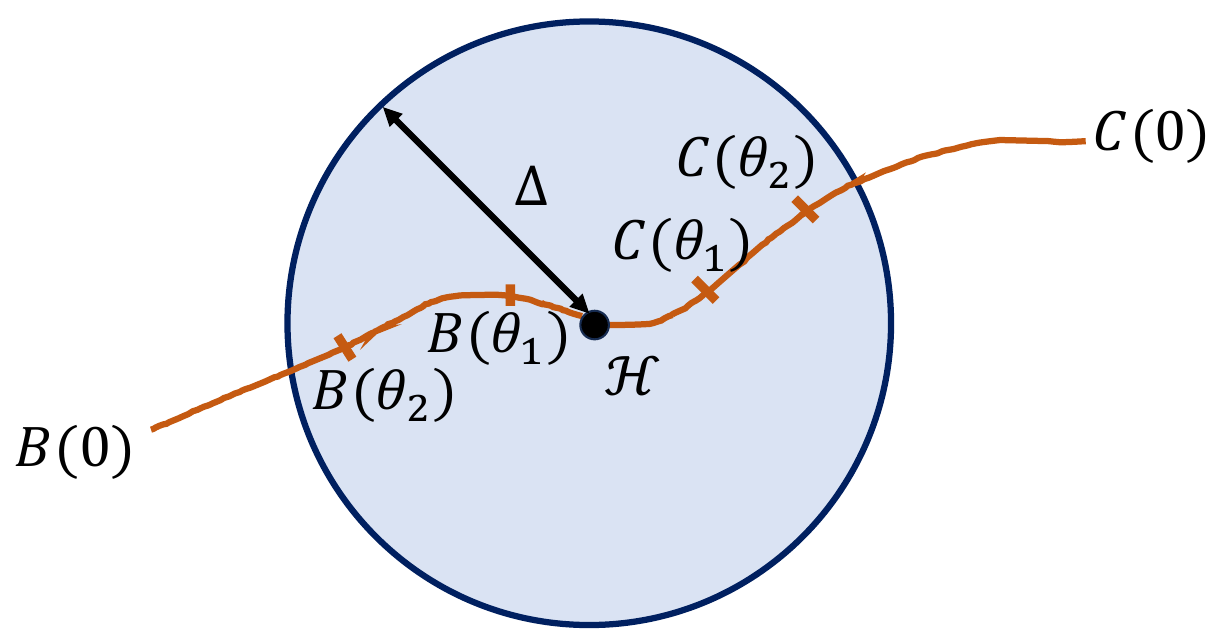}
\caption{The center of the ball denotes the Haar measure $\mathcal{H}$ and the circuits  $B(\theta_i)$ and $C(\theta_j)$ are instances of 
 small pullbacks of the Haar measure towards  the fixed circuits $B:=B(0)$ and $C:=C(0)$.}
\label{fig:security}
\end{figure}

Let $B_{1},\dots,B_{m}$ and $C_{1},\dots,C_{m}$ be the 
fixed set of gates corresponding to the circuits $B$ and $C$ respectively, which have the same architecture. Let $H_{1},\dots,H_{m}$ and $H'_{1},\dots,H'_{m}$ be local unitaries drawn from the Haar measure according to the circuit architecture of $B$ and $C$. 
 Recall that we can always find hermitian matrices $h_{k},h'_{k}$
such that $h_{k}=f^{-1}(H_{k})$ and $h'_{k}=f^{-1}(H'_{k}).$ We define the corresponding Cayley
paths 
\begin{eqnarray*}
B_{k}(\theta)  := & B_{k}f(\theta h_{k}) \quad
 \text{and}\quad
C_{k}(\theta)  := & C_{k}f(\theta h'_{k}).
\end{eqnarray*}
For a given $\theta$, we can instantiate the following circuits using the Cayley of the gates: 
\begin{eqnarray}\label{eq:BC_theta}
B(\theta):=\prod_{k=1}^{m}B_{k}(\theta),\quad
 \text{and}\quad C(\theta):=\prod_{k=1}^{m}C_{k}(\theta)\;.
 \end{eqnarray}

The distribution over a circuit with Haar gates is denote by $\mathcal{H_{A}}$
where $\mathcal{A}$ denotes the architecture that the circuit respects.
The distribution over a circuit $B(\theta)$ with architecture $\mathcal{A}$
is denoted by $\mathcal{H}_{\mathcal{A},\theta}$ which is the pullback
by $\theta$ by the Cayley path of the Haar measure where $|1-\theta|\le\Delta$.

Let $B(\theta_{1})\sim\mathcal{H}_{\mathcal{A},B,\theta}$
and $C(\theta_{1})\sim\mathcal{H}_{\mathcal{A},C,\theta}$ with probability
densities $\rho_{B}(U;\theta_{1})$ and $\rho_{C}(U;\theta_{1})$
where $U$ denotes the unitary dependence of the random
matrix distributions on the vector $U$. From ~\cite{movassagh2023hardness} and  Theorem~\ref{thm:no-go}, we have that the total variation distance between this distributions is
\begin{eqnarray*}
\text{TVD}\left( \rho_{B},\rho_{C}\right) = \Theta(m\Delta),
\end{eqnarray*}
which is used in the security of our protocol against a single malicious party.

Let $J_{B,\theta_{1},\theta_{2}}$ be the joint distribution on circuits
$\left(B(\theta_{1}),B(\theta_{2})\right)$ and similarly define $J_{C,\theta_{1},\theta_{2}}$
to be the joint distribution on circuits $\left(C(\theta_{1}),C(\theta_{2})\right)$.
We now show that the distance between $J_{B,\theta_{1},\theta_{2}}$  and $J_{C,\theta_{1},\theta_{2}}$ is $1$, which shows that 
our protocol is not secure if we have two or more colluding parties. 

To prove that, we first show this auxiliary lemma.

\begin{lemma}\label{lem:maxCorr}
Let $B(\theta)$ be the Cayley path for a circuit $B$. There exists a function $g_{B}$ such that $g_B(B(\theta_1),\theta_1,\theta_2) = B(\theta_2)$. Moreover, for $B \ne C$, $g_C(B(\theta_1),\theta_1,\theta_2) \ne B(\theta_2)$.
\end{lemma}

\begin{proof}
Let $g_B(U,\theta_1,\theta_2) = Bf(\frac{\theta_{2}}{\theta_{1}}f^{-1}(B^{\dagger}U))$. We have that
\begin{align*}
    g_B(B(\theta_1),\theta_1,\theta_2) &= Bf(\textstyle\frac{\theta_{2}}{\theta_{1}}f^{-1}(B^{\dagger}B(\theta_1)))     = Bf(\textstyle\frac{\theta_{2}}{\theta_{1}}f^{-1}(B^{\dagger}B f(\theta_1h_k))) 
    = Bf(\frac{\theta_{2}}{\theta_{1}}\theta_1h_k)) 
    =B(\theta_2).
\end{align*}

Moreover, since $f$ and $f^{-1}$ are bijections if $B \ne C$, we have that
\begin{align*}
    g_C(B(\theta_1),\theta_1,\theta_2) &= Cf(\textstyle\frac{\theta_{2}}{\theta_{1}}f^{-1}(C^{\dagger}B(\theta_1)))     = Cf(\textstyle\frac{\theta_{2}}{\theta_{1}}f^{-1}(C^{\dagger}B f(\theta_1h_k))) \ne B(\theta_2).
    \end{align*}
\end{proof}

We now prove the main lemma of this section.

\begin{lemma}
\label{lem:tvd}
$\text{TVD}\left(J_{B,\theta_{1},\dots,\theta_{k}},J_{C,\theta_{1},\dots,\theta_{k}}\right)=1$, for $k\ge 2$.
\end{lemma}
\begin{proof}
We show $k=2$ which is sufficient for general $k$ as the statistical distance does not decrease with growing $k$. First of all, let us define $J_{B,\theta_{1},\theta_{2}}(U_1,U_2)=\text{Pr}[B(\theta_1)=U_1,B(\theta_2)=U_2 ]$.  The TVD between $J_{B,\theta_{1},\theta_{2}}$  and $J_{C,\theta_{1},\theta_{2}}$ is defined given by 
\begin{eqnarray*}
\text{TVD}\left(J_{B,\theta_{1},\theta_{2}},J_{C,\theta_{1},\theta_{2}}\right) & = & \frac{1}{2}\int\left| J_{B,\theta_{1},\theta_{2}}(U_1,U_2)-J_{C,\theta_{1},\theta_{2}}(U_1,U_2) \right|\;dU_1 dU_2\\
 & = & \frac{1}{2}\int\left|p_{B}(U_1)\delta(U_2-g_B(U_1))-p_{C}(U_1)\delta(U_2-g_C(U_1))\right|\;dU_1 dU_2
\end{eqnarray*}
where by $g_B$ and $g_C$ are the functions from Lemma~\ref{lem:maxCorr}. From \Cref{lem:maxCorr}, these delta functions have  disjoint supports for $C\ne B$, and it follows that 
\begin{eqnarray*}
\text{TVD}\left(J_{B,\theta_{1},\theta_{2}},J_{C,\theta_{1},\theta_{2}}\right) & = & \frac{1}{2}\left(\int p_{B}(U_1)dU_1 + \int p_{C}(U_1)dU_1\right) = 1. 
\end{eqnarray*}
\end{proof}

\begin{corollary}\label{cor:impossibility-2qfss}
 The triple   $(\Gen, \Eval, \Rec)$, defined in Algorithms~\ref{algo:gen},~\ref{algo:eval},~\ref{algo:rec}, is {\em not}  a
     $\left(p, 2,  \ell, |C| \ell, \eta\right)$ QFSS scheme.
\end{corollary}
\begin{proof}
    The attacker chooses two distinct circuits $C^0$ and $C^1$. Given two share $C^b_i$ and $C^b_j$, one can perfectly distinguish if $b = 0$ or $b = 1$, since the total variation distance of the two distibutions is $1$ from \Cref{lem:tvd}.
\end{proof}

\appendix

\newpage
\bibliographystyle{unsrt}
\bibliography{refs} \end{document}